\documentclass[10pt]{article} 
\usepackage{amsmath, amsfonts, amsthm, amssymb, amscd}
\usepackage[mathcal]{euscript} 
\usepackage{dsfont, pxfonts, mathrsfs, verbatim}   
\usepackage{stmaryrd}
\usepackage{enumerate} 
\usepackage{graphicx}
\theoremstyle{plain}
\newtheorem{theorem}{Theorem}[section]
\newtheorem{proposition}[theorem]{Proposition} 
\newtheorem{lemma}[theorem]{Lemma}
\newtheorem{corollary}[theorem]{Corollary}

\numberwithin{equation}{section}

\newcommand \mut {\widetilde \mu}

\newcommand \lbrac \llbracket 
\newcommand \rbrac \rrbracket 
\newcommand \R    {\mathbb{R}}

\newcommand \RR {\mathbb{R}}
 
\newcommand \del  {\partial}  

\newcommand \be {\begin{equation}}
\newcommand \ee {\end{equation}}

\newcommand \lam \lambda
\newcommand \sig \sigma
\newcommand \gam \gamma

\newcommand \Hcal {\mathcal{H}}

\newcommand \Acal {\mathcal{A}}
\newcommand \Dcal {\mathcal{D}}
\newcommand \Ccal {\mathcal{C}} 
\newcommand \eps \epsilon
\newcommand \Lam \Lambda
\newcommand \BB {{\mathcal B}} 
\newcommand \la \langle
\newcommand \ra \rangle 

\let\oldmarginpar\marginpar
\renewcommand\marginpar[1]{\-\oldmarginpar[\raggedleft\footnotesize #1]%
{\raggedright\footnotesize  \bf #1}}

\begin{document} 

\title 
{On the area of the symmetry orbits
 in 
\\
weakly regular Einstein-Euler spacetimes with Gowdy symmetry
}
%
\author{Nastasia Grubic\footnote{
Laboratoire Jacques-Louis Lions \& Centre National de la Recherche Scientifique, 
Universit\'e Pierre et Marie Curie (Paris 6), 4 Place Jussieu, 75252 Paris, France.  
{\it Email:} grubic@ann.jussieu.fr, contact@philippelefloch.org.  
\newline 
{\it Key Words.} Einstein-Euler spacetime, Gowdy symmetry, cosmological spacetime, areal foliation, global geometry.  
\textit{\ AMS Class.}  Primary. 83C05,  35L65. Secondary. 76L05, 83C35. 
\newline SIAM Journal of Mathematical Analysis (2015). 
}  
\, and 
Philippe G. LeFloch$^*$ 
}  
\date{} 
\maketitle

\begin{abstract}  This paper establishes novel bounds for  Gowdy-symmetric Einstein-Euler spacetimes
and completes the analysis, initiated by LeFloch and Rendall, 
 of the global areal foliation for these spacetimes. We thus consider the initial value problem for the Einstein-Euler equations under the assumption of Gowdy symmetry. We establish that, for the maximal Cauchy development of future contracting initial data, the area of the group orbits approaches zero toward the future. This property holds as one approaches the future boundary of the spacetime, provided a geometry invariant associated with the Gowdy symmetry property is initially non-vanishing. Our condition is sharp within the class of spatially homogeneous spacetimes.  
\end{abstract}


\section{Introduction} 

In the past tventy years, siginificant progress has been made in the study of the Einstein equations under symmetry assumptions and, especially, the existence of global foliations for classes of (vacuum) spacetimes admitting two spacelike Killing fields was established. (Cf.~the textbook by Rendall \cite{II-Rendall-book}.) In particular, Moncrief \cite{II-Moncrief} treated Gowdy spacetimes while Isenberg and Weaver \cite{IsenbergWeaver} analyzed $T^2$ symmetric spacetimes: they were able to show the existence of a foliation by spacelike hypersurfaces on which the area of the symmetry orbits is constant and covers the whole range $(0, +\infty)$. We also refer to \cite{II-Andreasson, II-EardleyMoncrief, II-IsenbergMoncrief, II-Rendall1, II-Rendall-crush} for further results in this direction. 

More recently in \cite{BLSS,GrubicLeFloch,LeFloch,LeFlochRendall,LeFlochStewart1,LeFlochStewart2}, the second author together with collaborators  initiated the mathematical study of matter spacetimes described by the Einstein-Euler system for self-gravitating compressible fluids. Due to the existence of shock waves in solutions to the Euler system, the curvature of such spacetimes can only be defined in the sense of distributions, and it is particularly challenging to analyze their local, as well as their global, properties. In the present paper, we contribute to the existing theory of weakly regular Einstein-Euler spacetimes with Gowdy symmetry, initiated by LeFloch and Rendall \cite{LeFlochRendall}, and we establish novel bounds on solutions to the Einstein-Euler equations and consequently complete the mathematical analysis of the global areal foliation. We succeed here to determine the range of the area of the symmetry orbits which is found to be the whole interval $(0, +\infty)$, except for special solutions which are characterized geometrically. This solves a problem posed by LeFloch and Rendall in \cite{LeFlochRendall}. 
 
More precisely, we are interested in four-dimensional Lorentzian manifolds $(M, g)$ satisfying the Einstein equations 
\be
\label{EE1-II}
{G_\alpha}^\beta = {T_\alpha}^\beta, 
\ee
where $T_{\alpha\beta}$ denotes the stress-energy tensor of the fluid and $G_{\alpha \beta} := R_{\alpha\beta} - (R/2) g_{\alpha\beta}$ denotes the Einstein curvature tensor describing the geometry of the spacetime. Here, $R_{\alpha\beta}$ and $R$ denote the Ricci and scalar curvatures and the indices $\alpha, \beta$ vary from $0$ to $3$. 

The energy-momentum tensor of a perfect fluid is given by
\be
\label{stren-II}
{T_\alpha}^\beta := (\mu + p) \, u_\alpha u^\beta + p \, {g_\alpha}^\beta,
\ee
where $\mu >0$ is the mass-energy density of the fluid and $u$ its unit, timelike velocity vector. We assume the linear equation of state 
\be 
\label{eq:pressure-II} 
p:=k^2 \mu,
\ee
where $k\in (0,1)$ represents the sound speed in the fluid and does not exceed the speed of light (normalized to unity). The Bianchi identities imply  
\be
\label{EE2-II}
\nabla_\alpha T_\beta^\alpha=0,
\ee
which are nothing but the Euler equations describing the evolution of the fluid. 

We study the initial value problem for the Einstein--Euler equations when an initial data set is prescribed on a three-dimensional spacelike hypersurface $\Hcal$ whose topology coincides with the 3-torus $T^3$. In addition, we assume that the initial data are Gowdy-symmetric \cite{II-Gowdy},
that is, are invariant under the action of the Lie group $T^2$ and have vanishing twist constants.  

Our main result is as follows. 

\begin{theorem}[Future contracting Einstein-Euler spacetimes with Gowdy symmetry] 
\label{maintheo-II}  
Consider any Gowdy-symmetric initial data set with bounded variation (BV) regularity, defined on $T^3$ and associated with the Einstein-Euler equations, and assume that this initial data set has constant area $-t_0>0$ and is everywhere expanding toward the future. Then, there exists a BV regular, Gowdy symmetric spacetime $M$ satisfying the Einstein-Euler equations (defined in \eqref{EE1-II}, below) in the distributional sense, which is a future development of the given initial data set and is globally covered by a single coordinate chart with  
$$
M \simeq [t_0, t_1) \times T^3,   
$$
where the time variable $t$ is chosen to coincide with the minus the area of the symmetry orbits.  Furthermore, provided the geometric invariant $\Dcal$ associated with the Gowdy symmetry (see \eqref{def-Dcal}, below) is non-vanishing
$$
\Dcal \neq 0,
$$
one has 
$$
t_1 = 0
$$ 
and the area therefore approaches $0$ in the future. 
\end{theorem}  

The notion of BV regular spacetimes and the existence part of the above theorem were presented by the authors in \cite{GrubicLeFloch}. Our main contribution in the present paper is the fact that the area of the group orbits tends zero toward the future, as one approaches the future boundary of the spacetime.  We focus here on future contracting spacetimes, while future expanding spacetimes were already dealt with by LeFloch and Rendall \cite{LeFlochRendall}.  


\section{Einstein--Euler spacetimes}

\subsection{Formulation in areal coordinates}

We consider spacetimes $(M, g)$ that admit a foliation by a time function $t: M \to  I \in \R$, where $I$ is an interval. More precisely, we have 
$$
M = \bigcup_{t \in I} \Hcal_t,
$$
where each $\Hcal_t$ is a compact spacelike Cauchy hypersurface diffeomorphic to the initial hypersurface $\Hcal$ and $g^{\alpha\beta} \del_\alpha t$ is a future-oriented timelike vector field.  In Gowdy symmetry, it is natural to foliate the spacetime by the area function of the symmetry orbits. Its gradient $\nabla t$ is a timelike vector field, so this is always possible. In particular, either $\nabla t$ or $-\nabla t$ can be chosen to determine the time-orientation of $M$. Since discontinuous solutions of Euler equations are time-irreversible, the spacetime is  (uniquely) defined only in the future of the initial hypersurface, and we thus distinguish between two classes of initial data normalized so that:   
\be
\aligned
& \text{Future expanding spacetimes: } t_0 >0, 
\\
& \text{Future contracting spacetimes: } t_0 <0. 
\endaligned
\ee 
In the present paper, we are interested in describing the interval of existence $[t_0, t_1)$ for future contracting spacetimes. More precisely, we provide a  geometric condition that ensures $t_1 = 0$. In contrast, recall that, in expanding spacetimes, the time variable describes the whole interval $[t_0, +\infty)$.

In this section we state the field equations in areal coordinates and refer to \cite{LeFlochRendall} for the derivation. In areal coordinates, the metric reads 
\be
\label{areal-II}
g = e^{2(\eta-U)}(-a^2d t^2 + d\theta^2)+e^{2U}(dx + Ady)^2+e^{-2U}t^2dy^2, 
\ee
where the four metric coefficients $a,\eta, U, A$ (with $a>0$) depend upon the time variable $t$ and the spatial variable $\theta \in S^1 \simeq [0,1]$ 
(with periodic boundary conditions). 

It is convenient to replace  $\eta$ by a related metric coefficient $\nu$  and rescale the fluid density $\mu$, as follows
$$
\nu := \eta + \log (a), \qquad \mut:=e^{2(\nu-U)}\mu,
$$
so that the evolution equations for $U,A$ and $\nu$ follow from the Einstein equations \eqref{EE1-II} and read 
\be
\label{evolu1-II}
\aligned
\big( t^{-1} \, a^{-1}A_t \big)_t - \big( t^{-1} \, aA_\theta \big)_\theta 
&= -\frac{4}{at} \big( U_tA_t - a^2U_\theta A_\theta \big),
\\
\big( t \, a^{-1}(U_t - 1/(2t) \big)_t - \big( t \, a U_\theta \big)_\theta 
& = \frac{e^{4U}}{2ta} \big( A^2_t - a^2A^2_\theta \big), 
\\
\big( ta^{-1}(\nu_t + t\mut (1 - k^2)) \big)_t -(ta\nu_\theta)_\theta 
&= 2atU_\theta^2 + \frac{e^{4U}}{2at}A_t^2 + t a^{-1} \, \mut \, \frac{(1+k^2)}{1-v^2}. 
\endaligned
\ee
These equations are understood in the sense of distributions and are supplemented with three constraint equations 
\be
\label{constraint-II}
\aligned
&a_t = -at\mut (1 - k^2),
\\
&\nu_t  = t(U_t^2 + a^2U_\theta^2) + \frac{e^{4U}}{4t}(A_t^2 + a^2A_\theta^2)+ t\mut \, \frac{k^2 + v^2}{1-v^2},
\\
&\nu_\theta  = - 2t U_t U_\theta - \frac{e^{4U}}{2t} A_t A_\theta - a^{-1} t\mut \, \frac{(1+k^2)v}{1-v^2}. 
\endaligned
\ee

 On the other hand, using the divergence--free property of the energy--momentum tensor, we obtain the Euler equations   
\be
\label{fluid-II}
\aligned
&\del_t\Bigg(a^{-1}t\mut\frac{1+k^2v^2}{1-v^2}\Bigg)+\del_\theta\Bigg(t\mut\frac{(1+ k^2)v}{1-v^2}\Bigg) 
= 
a^{-1}t\mut(1-k^2) \, \Sigma_0, 
\\
&\del_t\Bigg(a^{-1}t\mut\frac{(1+k^2)v}{1-v^2}\Bigg) + \del_\theta\Bigg(t\mut\frac{k^2 + v^2}{1-v^2}\Bigg) 
= 
a^{-1}t\mut(1-k^2) \, \Sigma_1,
\endaligned
\ee
with the right-hand sides defined by 
$$
\aligned 
\Sigma_0  &:=  -\frac{k^2}{(1-k^2)t} - U_t + t(U_t^2 + a^2U_\theta^2) + \frac{e^{4U}}{4t}(A_t^2 + a^2A_\theta^2),
\\
\Sigma_1  &:= -aU_\theta + 2t \,aU_t U_\theta + \frac{e^{4U}}{2t} \,aA_t A_\theta.
\endaligned
$$ 
Our earlier analysis in \cite{GrubicLeFloch} has shown that it is sufficient to solve a subset of "essential equations'' first, consisting of the Euler equations for the fluid and the evolution equations for only $U, A, a$. Subsequently, $\nu$ is recovered from the remaining equations in \eqref{constraint-II}. 


\subsection{Energy functionals}

We will make use of the following two energy functionals
$$
\aligned
E_1(t) &:= \int_{S^1} h_1 \, d\theta, \qquad \quad  E_2(t) := \int_{S^1} \big( h_1 + h_1^M \big) \, d\theta, 
\endaligned
$$
with
$$
h_1 = a^{-1}\Bigg( \big(U_t - \frac{1}{2t}\big)^2 + a^2U_\theta^2  + \frac{e^{4U}}{4t^2}\big(A_t^2 + a^2 A_\theta^2\big)\Bigg),
\qquad
\quad 
h_1^M := a^{-1}\mut\frac{1+k^2v^2}{1-v^2}.
$$ 
Observe that $h_1$ can be obtained from the standard energy density by a transformation $U \mapsto U - \frac{1}{2}\log |t|$. To the energy densities $h_1$ and $h_1 + h_1^M$ we associate the corresponding fluxes $g_1$ and $g_1 + g_1^M$, respectively, defined by 
$$
g_1:= 2\big(U_t - \frac{1}{2t}\big)U_\theta  + \frac{e^{4U}}{2t^2}A_tA_\theta, 
\qquad 
\quad
g_1^M := a^{-1}\mut\frac{(1+k^2)v}{1-v^2}. 
$$
Precisely, we find 
$$
\aligned
\del_t h_1 + \del_\theta (ag_1) &= \frac{a_t}{a}h_1 - \frac{2}{t}\Bigg(\frac{(2tU_t - 1)^2}{4at^2} + \frac{e^{4U}}{4at^2}A_t^2\Bigg), 
\\
\del_t g_1 + \del_\theta (a h_1) &= \frac{a_t}{a}g_1 - \frac{g_1}{t}
\endaligned
$$ 
and
$$
\aligned
\del_t (h_1 + h_1^M) + \del_\theta (a(g_1 + g_1^M)) &= - \frac{1}{t}\Bigg(
\frac{(2tU_t - 1)^2}{2at^2} + \frac{e^{4U}A_t^2}{2at^2} + h_1^M + \alpha a^{-1}\mut\Bigg), 
\\
\del_t(g_1 + g_1^M) + \del_\theta (a (h_1 + \widehat{h}_1^M)) &= - \frac{1}{t}(g_1 + g_1^M).
\endaligned
$$
In particular, using the evolution equations, we obtain the following result. 

\begin{lemma}
The functionals $E_1$ and $E_2$ are monotonically increasing in the contracting direction and, specifically, one has 
\be
\frac{d E_1}{dt} (t) = \int_{S^1}
 \frac{a_t}{a}h_1 \, d\theta- \frac{2}{t}\int_{S^1}\big(\frac{(2tU_t - 1)^2}{4at^2} + \frac{e^{4U}}{4at^2}A_t^2\big)
\, d\theta ,
\label{der1-II}
\ee 
\be
\aligned
 \frac{d E_2}{dt} (t)  = -\frac{1}{t}\int_{S^1} \Bigg( 
\frac{(2tU_t - 1)^2}{2at^2} + \frac{e^{4U}A_t^2}{2at^2} + h_1^M + \alpha a^{-1}\mut
\Bigg)
\, d\theta,
\label{der2-II}
\endaligned
\ee
with $\alpha:= (3k^2 + 1)/4$.
\end{lemma}

Since $\frac{d E_2}{dt} \leq -\frac{2E_2}{t}$, it follows that 
$$
t^2E_1(t) \leq t^2E_2(t) \leq t_0^2E_2(t_0)
$$
and, therefore, $E_1(t)$ and $E_2(t)$ can not blow up before the singularity hypersurface $t=0$ is reached.


\subsection{Geometric invariants}

We introduce the following three {\sl geometric invariants}  
$$
\aligned
\alpha_+ &:= t^{-1}e^{4U}A_tA - (2t U_t - 1),
\\
\beta_+ &:= t^{-1}e^{4U}A_t,
\\
\gamma_+ &:=t^{-1}A_t(t^2 - e^{4U}A^2) + 2A\,(2t U_t - 1),
\endaligned 
$$
and their associated fluxes 
$$
\aligned
\alpha_- &:= t^{-1}ae^{4U}A_\theta A + 2atU_\theta,
\\
\beta_- &:= t^{-1}ae^{4U}A_\theta,
\\
\gamma_- &:= t^{-1}aA_\theta(t^2 - e^{4U}A^2) - 4atAU_\theta.
\endaligned 
$$
Strictly speaking, these quantities are not fully geometric, since they depend upon the choice of coordinates. We now observe that they obey explicit and simple transformations. 

\begin{lemma}[Conserved quantities for the Einstein--Euler equations] 
Given any BV regular solution to the Einstein--Euler equations, the associated functions  $\alpha_\pm, \beta_\pm, \gamma_\pm$ satisfy the balance laws  
$$
\aligned
&(a^{-1} \alpha_+)_t + (\alpha_-)_\theta = 0,
\\
&(a^{-1} \alpha_-)_t + (\alpha_+)_\theta =  \frac{1}{at} \alpha_-   - \frac{2}{at}\big(\alpha_-(\alpha_+ - A\beta_+) -  \alpha_+(\alpha_- - A\beta_-)\big), 
\\
&(a^{-1}\beta_+)_t + (\beta_-)_\theta = 0,
\\
&(a^{-1}\beta_-)_t + (\beta_+)_\theta = \frac{1}{at}\beta_- - \frac{2}{at}\big(\beta_-(\alpha_+ - A\beta_+) -  \beta_+(\alpha_- - A\beta_-)\big), 
\\
&(a^{-1}\gamma_+)_t + (\gamma_-)_\theta = 0,
\\
&(a^{-1}\gamma_-)_t + (\gamma_+)_\theta = \frac{1}{at}\gamma_- + \frac{2}{at}(t^2e^{-4U} + A^2)\big(\beta_-(\alpha_+ - A\beta_+) - \beta_+(\alpha_- - A\beta_-) \big). 
\endaligned
$$
\end{lemma}

It thus follows that the functionals 
$$
\Acal = \int_{S^1} a^{-1}\alpha_+, \quad  \BB = \int_{S^1} a^{-1}\beta_+, \quad \Ccal = \int_{S^1} a^{-1}\gamma_+, 
$$
are {\sl conserved in time.}

Observe next the following connection between the energy $E_1$ and a particular combination of the above geometric invariants: 
$$ 
\alpha_+^2 + \beta_+ \gamma_+ = (2tU_t - 1)^2+ e^{4U}A_t^2, \qquad \alpha_-^2 + \beta_- \gamma_- = 4t^2U_\theta^2+ e^{4U}A_\theta^2, 
$$
hence
$$
E_1(t) = \int_{S^1}\frac{1}{4at^2}\big((\alpha_+^2 + \beta_+ \gamma_+) + (\alpha_-^2 + \beta_- \gamma_-)\big) d\theta.
$$
In particular, in the spatially homogeneous case (analyzed in Section~\ref{homog_solutions}, below), we have 
\be
\label{204}
4a^{-1}t^2 E_1(t) = \Acal^2 + \BB\Ccal.
\ee 
The combination of the geometric invariants on the right-hand side of the above equation is somewhat special, as it remains unchanged under an action of isometries.  
This quantity
\be
\label{def-Dcal}
\Dcal := \Acal^2 + \BB\Ccal
\ee
plays a central role in the present wor. 

Depending upon the sign of $\Dcal$, it is more convenient to study certain combinations of the geometric invariants than others. In the following proposition, we follow Ringstr\"om~\cite{Ringstrom3} who treated vacuum spacetimes and 
we make use of certain isometries in order to impose specific values for $\Acal, \BB, \Ccal$.

\begin{proposition} \label{aux02}
The geometric invariants satisfy the following properties: 
\begin{itemize}

\item 
When $\Dcal\equiv \Acal^2 +\BB\Ccal>0$, then there is an isometry such that, if $\Acal_1, \BB_1$ and $\Ccal_1$ represent the conserved quantities of the transformed solution, then 
$$
\Acal_1 = - \sqrt{\Dcal}, \quad \BB_1 = \Ccal_1 = 0.
$$

\item 
When $\Dcal = 0$, then there is a transformation such that 
$$
\Acal_1 = \BB_1 = 0 \quad \text{ while }  C_1 = 0 \text{ or } C_1 = 1.
$$ 

\item 
Finally, when $\Dcal < 0$, it is obviously not possible to achieve $\BB_1 = 0$, however there exists a transformation such that 
$$
\Acal_1 = 0, \quad \BB_1 = -1, \quad \Ccal_1 = |\Dcal|. 
$$
\end{itemize}
\end{proposition}

Observe in passing that, in the spatially homogeneous case we always have $ \Acal^2 +\BB\Ccal \geq 0$. Moreover, if 
$\Acal^2 +\BB\Ccal = 0$, only the case $\Acal = \BB =  \Ccal = 0$ is possible.

\begin{proof}
We first consider several linear transformations of the Killing vector fields, denoted by $\xi$ and $\eta$, which leave the expression of the metric unchanged and we give the respective combinations needed to achieve the desired values of $\Acal, \BB, \Ccal$. 

\vskip.15cm

\noindent {\it Translations.} To achieve a translation of the metric coefficient $A$ by a constant $K$, we apply 
$$
(\xi, \eta) \mapsto (\xi, \, K \xi + \eta),
$$ 
which for the metric coefficients $A$ and $U$ implies 
$$
(A, \,U) \mapsto ( A + K,\, U),  
$$ 
whereas $a$ and $\nu$ remain unchanged. The conserved quantities $\Acal, \, \BB$ and $\Ccal$ change according to
$$
(\Acal, \,\BB,\,\Ccal) \mapsto (\Acal + K\BB,\, \BB, \,\Ccal - 2K\Acal - K^2 \BB).
$$

\vskip.3cm

\noindent{\it Dilations.} To obtain a dilation of $A$ by a constant factor $\lambda^2$, we apply
$$
(\xi, \eta) \mapsto (\frac{1}{\lambda}\xi, \, \lambda\eta),
$$ 
which yields   
$$
(A,\, U) \mapsto ( \lambda^2 A,\, U - \log \lambda),  \qquad (\Acal, \,\BB,\,\Ccal) \mapsto (\Acal,\, \frac{1}{\lambda^2}\BB, \,\lambda^2\Ccal ),
$$ 
to which we append $\nu \mapsto \nu + \log \lambda$.

\vskip.3cm

\noindent{\it Inversions.} Further, consider inverting the roles of $\xi$ and $\eta$ in the metric
$$
(\xi, \eta) \mapsto (\eta, \, -\xi),
$$ 
which, for the new metric functions, yields the following set of equations
$$
\aligned
e^{2\widehat{U}} &= e^{2U}A^2 + t^2e^{-2U},
\\
e^{2\widehat{U}}\widehat{A} &= - e^{2U}A,
\\
\widehat{a} &= a,
\\
e^{2(\widehat{\nu} - \widehat{U})} &= e^{2(\nu - U)},
\endaligned
$$
or in other words 
$$
(A,\, e^{-2U}) \mapsto (-\frac{A}{A^2 + t^2e^{-4U}},\, \frac{e^{-2U}}{A^2 + t^2e^{-4U}}),
  \qquad (\Acal, \,\BB,\,\Ccal) \mapsto -(\Acal,\,\Ccal, \,\BB).
$$ 

\vskip.3cm

\noindent{\it Reflections.} Finally, consider a reflection
$$
(\xi, \eta) \mapsto (\xi, \, -\eta),
$$ 
which results in   
$$
(A, \, U) \mapsto ( - A, \,U),  \qquad (\Acal, \,\BB,\,\Ccal) \mapsto (\Acal,\, -\BB, \,-\Ccal ), 
$$ 
while the coefficients $a$ and $\nu$ remain unchanged.  
 
\vskip.3cm

Furthermore, we observe that all of the above transformations leave the areal time $t^2 = g(\xi, \xi)g(\eta, \eta) - g(\xi, \eta)^2$ invariant. 

We are now in a position to conclude the proof. We first treat the case $\Dcal \geq 0$. Assume that $\BB \neq 0$. We can achieve $\BB = 0$, by carrying out a translation by $K = -\Acal/\BB + \sqrt{\Acal^2 +\BB\Ccal}/\BB$ and then an inversion. If $\Dcal = 0$ the desired value of $\Ccal$ can be achieved by applying a dilation and a reflection. If $\Dcal > 0$, we obtain $\Ccal=0$ applying another translation by $K = -\BB/(2\sqrt{\Acal^2 +\BB\Ccal})$.

If $\BB = 0$, then, in the case $\Dcal = 0$, we necessarily have $\Acal = 0$, hence by applying a dilation and/or a reflection we obtain the desired result. On the other hand, if $\Dcal > 0$, one first needs a translation to set $\Ccal$ to zero and then an inversion to change the sign of $A$ if necessary.

Assume now that $\Dcal < 0 $. If $\Acal = 0$, we first apply a reflection if $\BB$ is positive and then a dilation by $\lambda = \sqrt{|\BB|}$. If $\Acal \neq 0$, we first apply a translation by $K = -\Acal/\BB$ and then proceed as above. This completes the proof of Proposition~\ref{aux02}.  
\end{proof}


\subsection{Spatially homogeneous Einstein-Euler spacetimes}
\label{homog_solutions}

Homogeneous solutions to the Einstein-Euler system are characterized by the vanishing of the space derivative of the metric coefficients $U,A, a$ and fluid variables $\mu, v$.  
The momentum constraint equation then implies
$v=0$  
and the system simplifies to 
\be
\label{ode_h} 
\aligned
\big(a^{-1}t(U_t - 1/(2t))\big)_t &=  \frac{e^{4U}}{2at}A_t^2, 
\\
\big(a^{-1}t^{-1}e^{2U}A_t \big)_t &= - 2\frac{e^{2U}}{a\, t} U_t A_t,  
\\
a_t &= -at\mut (1-k^2),
\\
\big(a^{-1}t\mut\big)_t &= a^{-1} t\mut(1-k^2) \, \big(-\frac{1}{(1-k^2)}\frac{\alpha}{t} + atE_1(t) \big),
\endaligned
\ee 
where we recall that $\alpha = (3k^2 + 1)/4$. It is natural to consider \eqref{ode_h} as a first-order system in the variables $a, U_t, A_t$ and $\mu$.  Furthermore, the energy $E_1$ now satisfies
\be
\label{energy1}
\frac{d}{dt}(a^{-1}t^2E_1(t)) = 0,
\ee
as follows from \eqref{204}. As it turns out, it is this energy that determines the behavior of the system.  Indeed, if the energy does not initially vanish, then the solution exists on the whole interval $[t_0, 0)$. On the other hand, if the energy vanishes, the situation is more complex and the evolution depends critically upon the initial fluid density. We collect these observations in the following theorem. To facilitate the exposition, we introduce the normalized density 
\be
m := \frac{4}{3}t^2\mut.
\ee

\begin{theorem}[Spatially homogeneous Einstein-Euler spacetimes] 
\label{407h} 
Given any initial data $u_0$ and the corresponding solution $u = (U_t, A_t, a, m)$ 
of the ordinary differential equations \eqref{ode_h}, one can distinguish between the following cases: 
\begin{description}

\item[(A)] If $E_1(t_0) \neq 0$, then the coefficient $a$ and consequently the functions $tU_t, A_t$ and $m$ remain globally bounded and, moreover, $m \rightarrow 0$ as $t \rightarrow 0$.

\item[(B)] If $E_1(t_0) = 0$, then one has 
$$
U_t=\frac{1}{2t},  \qquad A_t=0,  \qquad  a^{-1}m \big(\frac{t_0}{t}\big)^{1 - \alpha}  =  a_0^{-1}m_0,
$$ 
and, depending upon the initial value $m_0 = m(t_0)$, the following three subcases may arise: 

\begin{itemize} 

\item[(i)] When $m_0$ is sufficiently small and, specifically, 
$$
m_0<1,
$$
then $a$ remains bounded and $m \rightarrow 0$ as $t\rightarrow 0$ and, in particular, $\mut\rightarrow + \infty$. 

\item[(ii)] When $m_0$ is unit, that is, 
$$
m_0=1,
$$
then $m = 1$ is constant and $a(t) = a_0\big(\frac{t_0}{t}\big)^{1 - \alpha}$.

\item[(iii)] When $m_0$ is sufficiently large and, specifically, 
$$
m_0>1,
$$
then the function $a$ and, consequently, the density $m$ {\rm blow--up} for some non-vanishing $t_1 \in (t_0, 0)$ and, in particular, the density $\mut$ blows--up. 
\end{itemize}
\end{description}
\end{theorem}

This theorem clarifies the relationship between the blow-up condition in (B) and the non-vanishing energy condition in (A). The energy acts to diminish $m$ in such a way that $m$ is finally always less than 1. 

It is convenient to introduce the change of variable 
\be
\label{time}
\tau = -\log\big(\frac{t}{t_0}\big)^{1-\alpha}, \qquad \tau\in[0, +\infty), 
\ee
where $1-\alpha=\frac{3}{4}(1-k^2)$.  
We rewrite the last two equations in \eqref{ode_h} in terms of the unknowns $(a, m)$, as functions of the new time variable $\tau$, and obtain the autonomous system
\be
\label{sys1}
\aligned
\frac{d(a^{-1})}{d\tau} &= -a^{-1}\,m,
\\ 
\frac{d(a^{-1}m)}{d\tau} &= -C_0a\,m   - a^{-1}\,m,
\endaligned
\ee
 where, thanks to \eqref{energy1}, we have set
\be
\label{cnst-energy}
C_0 := \frac{4}{3}t_0^2E_1(t_0)a_0^{-1}. 
\ee 

\begin{proof}[Proof of Theorem \ref{407h}]
Substracting the two equations in \eqref{sys1}, we obtain  $\frac{\,d}{d\tau}(a^{-1}(m - 1) + C_0\,a) =0$, thus 
\be
\label{cnst} 
a^{-1}(m - 1) + C_0\,a = C_0\,a_0 + X_0, 
\ee
where $X_0 := a_0^{-1}(m_0 - 1)$, and therefore
\be
\label{eq_a}
\frac{da}{d\tau} =   -a(C_0\,a^2 - (C_0\,a_0 + X_0)\,a - 1).
\ee
The right-hand side is a polynomial of degree three with roots $a_-, 0, a_+$ given by 
\be
\label{roots}
a_\pm = \frac{(C_0a_0 + X_0) \pm \sqrt{(C_0a_0 + X_0)^2 + 4C_0}}{2C_0}. 
\ee 
This polynomial is positive on the interval $(0, a_+)$ and, since $a_0 \in (0, a_+)$, we have
$$
\aligned
\lim_{\tau\rightarrow + \infty} a(\tau) &= a_+,
\\
\lim_{\tau\rightarrow +\infty} a^{-1}(\tau)\big(m(\tau) - 1\big) &= C_0\,a_- <0,
\\
\lim_{\tau\rightarrow +\infty} m(\tau) &= C_0\,a_-a_+ + 1 = 0,
\endaligned
$$
the last equality following from Vieta's formula.

In the present notation, the special case $E(t_0) = 0$ is equivalent to $C_0 = 0$. The equation $\eqref{cnst}$ now simplifies to
$a^{-1}(\tau)\big(m(\tau) - 1\big) = X_0$ thus $m(\tau) = a(\tau)X_0 + 1$. We then have 
$\frac{da}{d\tau} =   X_0\,a^2 + a$
and therefore $X_0a(\tau) = \frac{\frac{a_0X_0}{a_0X_0 +1 }e^{\tau}}{1 - \frac{a_0X_0}{a_0X_0 +1 }e^{\tau}}$, so  
$$
X_0a(\tau) = m(\tau)\frac{a_0X_0}{a_0X_0 +1 }e^{\tau}. 
$$ 
Hence, depending on $X_0$, we have three cases:
\begin{itemize}

\item When $X_0>0$, both $a(\tau)$ and $m(\tau)$ are monotonically increasing and blow up for a finite $\tau$. The critical time $\tau_c$ depends on the initial data, as follows:  
$$
\tau_c = \log \frac{m_0}{m_0 -1} = \log \frac{a_0X_0 + 1}{a_0X_0} .
$$

\item When $X_0=0$, we have 
$$
\aligned
a(\tau) &= a_0e^\tau,
\qquad
m(\tau) = 1.
\endaligned
$$

\item When $X_0<0$, the variable $m( \tau)$ is monotonically decreasing and we have 
$$
\aligned
\lim_{\tau\rightarrow +\infty} a(\tau) &= \frac{1}{|X_0|},
\qquad 
\lim_{\tau\rightarrow +\infty} m(\tau) = 0.
\endaligned
$$ 
The above explicit solutions for $a(\tau), m(\tau)$ imply that $\lim_{\tau\to +\infty}\mut(\tau) = +\infty$.
\end{itemize}
\end{proof}


\section{The areal foliation of future contracting spacetimes}

\subsection{An upper bound on the fluid variables}

We consider any BV-regular solution to the Einstein-Euler system, defined on some interval $[t_0, t_c)$, and we establish that
$\sup_{S^1} a$ is controled by the lower bound of the function $t \mapsto \int_{S^1}a^{-1}(t, \theta) \, d\theta$.

\begin{proposition}
\label{prop1-II}
If there exists a positive function $C=C(t)$ (for $t \in [t_0, t_c)$) such that  
\be
\label{keyassumption-II}
C(t) \leq \int_{S^1}a^{-1}(t, \theta) \, d\theta, \qquad t \in [t_0, t_c), 
\ee
then there exists a function $C_1=C_1(t)$ such that 
\be
\label{keyproperty-II}
\sup_{S^1}\int_{t_0}^t\frac{\mut}{1 - v^2}(\tau, \cdot) \, d\tau \leq C_1(t), \qquad t \in [t_0, t_c).  
\ee
A~fortiori, the function $\sup_{S^1}\int_{t_0}^t \mut(\tau, \cdot) \, d\tau$ is bounded and, consequently, the metric coefficient $a$ is uniformly bounded on $S^1$. 
\end{proposition}

This statement provides us the conclusion in Theorem~\ref{maintheo-II}, under the assumption~\eqref{keyassumption-II}. The latter will be connected to our geometric invariant $\Dcal$ in Proposition~\ref{prop:2-II}, below. 

To establish Proposition~\ref{prop1-II}, we rely on the following two lemmas about 
 the time--average of the mass density $\int_{t_0}^t\frac{\mut}{1 - v^2}(\tau, \cdot) \, d\tau$ and its variation in space, respectively.  

\begin{lemma}[Lower bound for the averaged mass--energy density] 
\label{lem:1-II}
For all $t\in [t_0, t_c)$, one has 
$$
(1 + k^2) \, \Bigg(\int_{S^1}a^{-1}(t, \theta)d\theta\Bigg) \, \inf_{S^1}\int_{t_0}^t\frac{\mut|\tau|}{1 - v^2}(\tau, \cdot) \, d\tau 
\leq
\frac{1}{1 - \alpha}\Bigg(1 - \big(\frac{t}{t_0}\big)^{1 - \alpha}\Bigg)\int_{S^1} |t_0|^2 h^M_1(t_0, \cdot ) \, d\theta.
$$
\end{lemma}

\begin{proof} Namely, from the energy--type estimate \eqref{der2-II} we deduce
$$
\int_{t_0}^t \int_{S^1}|\tau| h_1^M(\tau, \cdot)\,d\theta\,d\tau 
\leq 
 \int_{t_0}^t \Bigg(\frac{t_0}{\tau}\Bigg)^\alpha \int_{S^1} |t_0| h^M_1(t_0, \cdot ) \, d\theta\, d\tau.
$$
Since $a^{-1}$ is monotonically decreasing in time, we obtain 
$$
\aligned
& \Bigg( \int_{S^1}a^{-1}(t, \theta)d\theta \Bigg) \,\inf_{S^1}\int_{t_0}^t\frac{\mut |\tau|}{1 - v^2}(\tau, \cdot) \, d\tau  
\\
&\leq \int_{S^1}a^{-1}(t, \theta) \int_{t_0}^t \frac{\mut |\tau|}{1 - v^2}(\tau, \theta)\,d\tau \, d\theta 
 \leq \int_{t_0}^t \int_{S^1} a^{-1} \mut |\tau|\frac{1 + k^2v^2}{1 - v^2}(\tau, \theta)\,d\theta\,d\tau. 
\endaligned
$$
\end{proof}

\begin{lemma}[Variation in space of the averaged mass--energy density] 
\label{prop:2lem}
For all $t\in [t_0, t_c)$ and $\theta_0, \theta_1\in S^1$, one has 
$$
\aligned
& k^2\int_{t_0}^{t} \frac{\mut \, |\tau|}{1 - v^2}(\tau, \theta_1)\,d\tau  
\\
& \leq (1 + k^2) \hskip-.1cm \int_{t_0}^{t} \frac{\mut |\tau|}{1 - v^2}(\tau, \theta_0) \,d\tau
  + |t_0|( E_1(t) - E_1(t_0)) +  \Bigg(1 + \big(\frac{t_0}{t}\big)^{\alpha}\Bigg)\int_{S^1} |t_0| h^M_1(t_0, \cdot ) \, d\theta. 
\endaligned
$$
\end{lemma}

\begin{proof} Setting $M:= \frac{\mut}{1-v^2}>0$, integrating the ``second'' Euler equation in \eqref{fluid-II} over some slab 
$[t_0, t]\times[\theta_0, \theta _ 1]$, and finally using $|v|<1$, we obtain (recalling that the time variable takes negative values)
$$
\aligned 
&  (1 + k^2)\int_{t_0}^{t} \tau M(\tau, \theta_1) \, d\tau 
- k^2 \int_{t_0}^{t} \tau M(\tau, \theta_0) \, d\tau 
\\
&\leq  \int_{t_0}^{t} \tau M\, (k^2 + v^2)(\tau, \theta_1) \, d\tau 
           - \int_{t_0}^{t} \tau M\, (k^2 + v^2)(\tau, \theta_0) \, d\tau
\\
& = \int_{t_0}^{t}\int_{\theta_0}^{\theta_1} \tau \frac{a_\tau}{a}g_1 \, d\theta d\tau 
          + (1 + k^2) \int_{\theta_0}^{\theta_1} \Big( t a^{-1} Mv(t, \theta)   - t_0 a^{-1} Mv(t_0, \theta) \Big) \, d\theta.
\endaligned
$$
In the above identity, the double integral term is controlled by $E_1$, i.e., 
$$
\aligned
 \int_{t_0}^{t}\int_{\theta_0}^{\theta_1} \frac{a_\tau}{a}\tau g_1 d\theta d\tau 
&\leq  - \int_{t_0}^{t}\int_{S^1} \frac{a_\tau}{a}\tau h_1 d\theta d\tau 
 \leq |t_0| \, \big( E_1(t) - E_1(t_0) \big),
\endaligned
$$
whereas, for the other two terms of the right-hand side, 
$$
\aligned
& (1 + k^2)  \int_{\theta_0}^{\theta_1} \Big( t a^{-1} Mv(t, \theta)  - t_0 a^{-1} Mv(t_0, \theta) \Big) \, d\theta
\\
& \leq \int_{S^1}|t| h_1^M(t, \cdot)\,d\theta + \int_{S^1}|t_0| h_1^M(t_0, \cdot)\,d\theta 
 \leq  \Bigg(1 + \big(\frac{t_0}{t}\big)^{\alpha}\Bigg)\int_{S^1} |t_0| h^M_1(t_0, \cdot ) \, d\theta.
\endaligned
$$
Hence, we obtain  
$$
\aligned
& - k^2 \int_{t_0}^{t} \tau M(\tau, \theta_1)d\tau
\\
& \leq - (1 + k^2)\int_{t_0}^{t} \tau M(\tau, \theta_0) \, d\tau + |t_0|( E_1(t) - E_1(t_0)) +  \Bigg(1 + \big(\frac{t_0}{t}\big)^{\alpha}\Bigg)\int_{S^1} |t_0| h^M_1(t_0, \cdot ) \, d\theta,
\endaligned
$$
which is the desired estimate. 
\end{proof}

\begin{proof}[Proof of Proposition~\ref{prop1-II}]
Combining Lemmas~\ref{lem:1-II} and \ref{prop:2lem}, we obtain  
$$
\aligned 
&  k^2 \sup_{S^1}\int_{t_0}^{t} \frac{\mut \, |\tau|}{1 - v^2}(\tau, \cdot) \, d\tau 
\\
& \leq (1 + k^2)\inf_{S^1} \int_{t_0}^{t}  \frac{\mut \, |\tau|}{1 - v^2}(\tau, \cdot) \, d\tau 
   + |t_0|( E_1(t) - E_1(t_0)) +  \Bigg(1 + \big(\frac{t_0}{t}\big)^{\alpha}\Bigg)\int_{S^1} |t_0| h^M_1(t_0, \cdot ) \, d\theta
\\
&\leq C_0(t) \, \Bigg(\int_{S^1}a^{-1}(t, \theta)d\theta\Bigg)^{-1} + C_1(t)
\endaligned
$$ 
for some (explicitly computable) functions $C_0(t), C_1(t)>0$. 
\end{proof}


\subsection{A lower bound on the geometry variables}

It remains to establish that 
the function $\int_{S^1}a^{-1}d\theta$ remains bounded below,  provided the geometric invariant is non-vanishing. 

\begin{proposition}
\label{prop:2-II}
Consider a spacetime with non-vanishing geometric invariant $\Dcal = \Acal^2 + \BB \Ccal\neq 0$. 
 Then, for some smooth function $F=F(s)$ vanishing at the origin $s=0$, one has 
\be
\label{eq304-II}
|\Dcal| \leq t^2 F\Big(  E_1(t)\int_{S^1} a^{-1}(t, \theta)\, d\theta \Big),  
\ee 
and, in particular, the function $t \in [t_0, t_c) \mapsto \int_{S^1} a^{-1}d\theta$ does not attain zero unless $t_c=0$ and $t \to t_c$.
\end{proposition}

We will actually prove that $F$ involves polynomial and exponential factors, only, which will lead us to the following result.

\begin{corollary} 
Consider a spacetime with non-vanishing geometric invariant $\Dcal = \Acal^2 + \BB \Ccal\neq 0$. Then, one has 
$t \in [t_0, 0)$ and, as $t \to 0$,  
$$
E_1(t) \, \int_{S^1} a^{-1}(t, \theta) \, d\theta \to + \infty
$$
so that, in particular, 
$$
E_1(t)  \to + \infty, 
\qquad 
 \frac{1}{t^2} \int_{S^1}a^{-1}d\theta \to +\infty. 
$$ 
\end{corollary}

We begin with several auxiliary results needed for the proof of Proposition \ref{prop:2-II}. We introduce the mean values of the metric coefficients $A$ and $U$, defined by 
$$
<A> := \int_{S^1}A\, d\theta, \qquad <U> := \int_{S^1}U\, d\theta.
$$
From the definition of $\Acal, \BB, \Ccal$ and by a straightforward calculation, the following result is immediate. 

\begin{lemma}
\label{aux03}
The geometric invariants $\Acal, \BB$ and $\Ccal$ satisfy 
\be
\label{eq301}
\BB<A>  = \Acal + 2t\int_{S^1} a^{-1} (U_t - 1/2t) \,d\theta + \int_{S^1}a^{-1}t^{-1}A_te^{4U}\big( A - <A> \big) d\theta,
\ee 
\be
\label{eq302}
t^2\BB e^{-2<U> } = t\int_{S^1}a^{-1}A_te^{2U}e^{2(U - <U>) }\, d\theta,
\ee
and 
\be
\label{eq303}
\aligned
\Ccal + 2\Acal<A>  -\BB<A> ^2 
= \, 
& t\int_{S^1}a^{-1}A_t \, d\theta - \int_{S^1}a^{-1}t^{-1}A_te^{4U}\big(A - <A>\big)^2\, d\theta
\\
& -4t\int_{S^1} a^{-1} (U_t - 1/2t) \big(A -<A> \big)\,d\theta.   
\endaligned 
\ee
\end{lemma}

\begin{lemma}
\label{aux01}
The following estimates hold: 
\be\label{eq305}
 \int_{S^1}|U_\theta|d\theta \leq E_1(t)^{1/2}\big(\int_{S^1}a^{-1}d\theta\big)^{1/2},
\ee
\be\label{eq305a}
\int_{S^1}a^{-1}|U_t - 1/2t|d\theta \leq E_1(t)^{1/2}\big(\int_{S^1}a^{-1}d\theta\big)^{1/2},
\ee
\be\label{eq306}
\int_{S^1}e^{2U}|A_\theta|d\theta \leq 2|t|E_1(t)^{1/2}\big(\int_{S^1}a^{-1}d\theta\big)^{1/2},
\ee
\be\label{eq306a}
\int_{S^1}a^{-1}e^{2U}|A_t|d\theta \leq 2|t|E_1(t)^{1/2}\big(\int_{S^1}a^{-1}d\theta\big)^{1/2},
\ee
\be\label{eq307}
e^{2U}\big(A - <A>\big) \leq \big(\int_{S^1} e^{2U}|A_\theta|\,d\theta\big) \mathrm{exp}\big(2\int_{S^1}|U_\theta|d\theta\big),
\ee
where the last inequality remains true if $U$ is replaced by its mean $<U>$.
\end{lemma}
\begin{proof}
The first four inequalities are straightforward, whereas for the last one we can write 
$$
\aligned
e^{2U}\big(A - <A>\big)   &= \int_{S^1} e^{2U(t, \theta')}\big(A(t, \theta') - A(t, \theta'')\big)\,d\theta''
\\
&\leq \int_{S^1}\int_{\theta'}^{\theta''} e^{2U(t, \theta')}|A_\theta(t, \theta)|\,d\theta 
\\
&
 = \int_{S^1}\int_{\theta'}^{\theta''} e^{2(U(t, \theta') - U(t, \theta))}e^{2U(t, \theta)}|A_\theta(t, \theta)|\,d\theta
\endaligned
$$
thus
$$
\aligned
e^{2U}\big(A - <A>\big)   
&\leq  \int_{S^1} \int_{\theta'}^{\theta''} e^{2\int_{S^1}|U_\theta(t, \theta''')|d\theta'''}e^{2U(t, \theta)}|A_\theta(t, \theta)|\,d\theta
\\
&\leq   \big(\int_{S^1} e^{2U(t, \theta)}|A_\theta(t, \theta)|\,d\theta\big) \mathrm{exp}\big(2\int_{S^1}|U_\theta(t, \theta)|d\theta\big)
\endaligned 
$$
and, similarly, with $U$ replaced by $<U>$.
\end{proof}

\begin{lemma}
\label{aux04}
Provided $\BB\neq 0$, then the averages $<A>$ and $e^{-2<U>}$ do not blow-up. In addition, if $\Dcal<0$ then neither does $e^{2<U>}$.
\end{lemma}

\begin{proof}
To show that $e^{-2<U>}$ does not blow-up, we apply Lemma \ref{aux01} and Holder's inequality to the equation \eqref{eq302}. Similarly, we use the equation \eqref{eq301} in order to check the claim for $<A>$. If $\Dcal < 0$, we can multiply by $e^{2<U>}$, then re-arrange the equation \eqref{eq303} and thus obtain
\be
\label{eq308}
\aligned
& \BB e^{2<U>}\Bigg(\big(<A> - \frac{\Acal}{\BB}\big)^2 - \frac{\Acal^2 +\BB\Ccal}{\BB^2}\Bigg) 
\\
&=
 \int_{S^1}a^{-1}t^{-1}A_te^{2U}e^{2(U + <U>)}\big(A - <A>\big)^2\, d\theta 
\\
&\quad +2\int_{S^1} a^{-1} (2tU_t - 1) e^{2<U>}\big(A -<A>\big)\,d\theta  
- \int_{S^1}a^{-1}tA_te^{2<U>}\, d\theta.
\endaligned
\ee
On the left-hand side, both terms have the same sign, so the claim follows from Lemma \ref{aux01} and Holder's inequality.
\end{proof}

\begin{proof}[Proof of Proposition \ref{prop1-II}] 
For ease in the presentation, it is convenient to define 
$$
\widehat{E}_1(t) :=4 E_1(t) \int_{S^1} a^{-1}d\theta.
$$
We first consider the case $\Dcal>0$. According to Proposition~\ref{aux02}, we can assume $\BB = 0$ and $\Acal = - \sqrt{\Dcal}$, hence \eqref{eq301} simplifies and reads 
$$
\sqrt{\Dcal} = 2t\int_{S^1} a^{-1} (U_t - 1/2t) \,d\theta + \int_{S^1} a^{-1}t^{-1}A_te^{4U}\big( A - <A> \big) \,d\theta.
$$
By Lemma \ref{aux01}, we have
$$
\big| 2t\int_{S^1} a^{-1} (U_t - 1/2t) \,d\theta  \big| \leq  |t| \widehat{E}_1(t)^{1/2} 
$$
and 
$$
\big| \int_{S^1} a^{-1}t^{-1}A_te^{4U}\big( A - <A> \big) \,d\theta \big| \leq  |t| \, \widehat{E}_1(t)\, \mathrm{exp}\widehat{E}_1(t)^{1/2},
$$
hence
$$
\Dcal 
\leq 
 |t|^2 \widehat{E}_1(t)
\Bigg(1 + \widehat{E}_1(t)^{1/2} \, \mathrm{exp}\widehat{E}_1(t)^{1/2}\Bigg)^2.
$$

Assume next that $\Dcal<0$. Again by Proposition~\ref{aux02}, we may assume $\Acal = 0$, $\BB = -1$, $\Ccal = |\,\Dcal\,|$, in which case the identity \eqref{eq303} reads
$$
\aligned
|\,\Dcal\,| = \, 
& t\int_{S^1}a^{-1}A_t \, d\theta - \int_{S^1}a^{-1}t^{-1}A_te^{4U}\big(A - <A>\big)^2\, d\theta
\\
& 
 -4t\int_{S^1} a^{-1} (U_t - 1/2t) \big(A -<A> \big)\,d\theta.   
\endaligned 
$$
Using Lemma \ref{aux01}, we obtain the estimates
$$
\big| t\int_{S^1}a^{-1}A_t \, d\theta \big| 
\leq 
t^2 e^{- 2<U>}  \widehat{E}_1(t)^{1/2} \,\mathrm{exp} \widehat{E}_1(t)^{1/2},
$$
$$
\big| \int_{S^1}a^{-1}t^{-1}A_te^{4U}\big(A - <A>\big)^2\, d\theta \big| 
\leq 
t^2 e^{- 2<U>}\widehat{E}_1(t)^{3/2}\,\mathrm{exp}\, 2\widehat{E}_1(t)^{1/2},
$$
and
$$
\big| 2t\int_{S^1} a^{-1} (U_t - 1/2t) \big(A -<A> \big)\,d\theta \big| \leq  t^2 e^{- 2<U>} \widehat{E}_1(t) \, \mathrm{exp}\widehat{E}_1(t)^{1/2}. 
$$
Hence, the claim follows if $e^{- 2<U>}$ is bounded from above by an expression containing the ``right power'' of $\widehat{E}_1(t)$. However,  using the estimates from Lemma  \ref{aux01} for the equation \eqref{eq302}, it is straightforward to conclude that  
$$
e^{- 2<U>} \leq  \widehat{E}_1(t)^{1/2}\,\mathrm{exp}\widehat{E}_1(t)^{1/2}.
$$
Collecting the conclusions above, we obtain \eqref{eq304-II}. 
\end{proof}


\section*{Acknowledgments}

This paper was completed when the second author (PLF) spent the Fall Semester 2013 at the Mathematical Sciences Research Institute (MSRI), Berkeley, thanks to some financial support from the National Science Foundation under Grant No. 0932078 000. 
The authors were also supported by the Agence Nationale de la Recherche through the grants ANR 2006-2--134423 (Mathematical Methods in General Relativity) and ANR SIMI-1-003-01 (Mathematical General Relativity.~Analysis and geometry of spacetimes with low regularity).


\end{document}